%% file: spcospark_conf.tex
\documentclass[conference]{IEEEtran}   	
\usepackage{graphicx}										
\usepackage{amssymb}
\usepackage{amsmath}
\usepackage{amsthm}
\usepackage[utf8]{inputenc}
\usepackage{tikz}		
\usetikzlibrary{positioning,chains,fit,shapes,calc}
\usepackage{algpseudocode} 
\usepackage{algorithm}
\usepackage{amsthm}
\usepackage{comment}
\usepackage{cite}

\input{defs}

\newtheorem{Theorem}{Theorem}
\newtheorem{lemma}{Lemma}
\newtheorem{corollary}{Corollary}
\newtheorem{definition}{Definition}

\renewcommand{\baselinestretch}{.99}

\begin{document}
\title{Generic Cospark of a Matrix Can Be Computed in Polynomial Time}


\author{\IEEEauthorblockN{Sichen Zhong\IEEEauthorrefmark{1} and 
Yue Zhao\IEEEauthorrefmark{2}}
\IEEEauthorblockA{\IEEEauthorrefmark{1}Department of  Applied Mathematics and Statistics, \IEEEauthorrefmark{2}Department of  Electrical and Computer Engineering \\Stony Brook University, Stony Brook, NY, 11794, USA\\
Emails: \{sichen.zhong, yue.zhao.2\}@stonybrook.edu
}
}


\maketitle

\begin{abstract}
The cospark of a matrix is the cardinality of the sparsest vector in the column space of the matrix. Computing the cospark of a matrix is well known to be an NP hard problem. Given the sparsity pattern (i.e., the locations of the non-zero entries) of a matrix, if the non-zero entries are drawn from independently distributed continuous probability distributions, we prove that the cospark of the matrix equals, with probability one, to a particular number termed the \emph{generic cospark} of the matrix. The generic cospark also equals to the \emph{maximum} cospark of matrices consistent with the given sparsity pattern. We prove that the generic cospark of a matrix can be computed in polynomial time, and offer an algorithm that achieves this. 
\end{abstract}

\section{Introduction} \label{sec:intro}

The \emph{cospark} of a matrix $A \in \mathbb{R}^{m \times n}, m>n$\footnote{We note that the results in this paper can be straightforwardly generalized to complex numbers.}, denoted by $\textit{cospark($A$)}$, is defined to be the \emph{cardinality of the sparsest vector in the column space of $A$} \cite{candes2005decoding}. In other words,  $\textit{cospark($A$)}$ is the optimum value of the following $l_0$-minimization problem: 
\begin{align}
\underset{x}{\text{minimize}} &~~ ||Ax||_0,  \label {cosparkprob}\\
\text{subject to} &~~ x \neq 0, \nn
\end{align}
where $||Ax||_{0}$ is the number of nonzero elements in the vector $Ax$. 
It is well known that solving \eqref{cosparkprob} is an NP-hard problem. Indeed, it is equivalent to computing the \emph{spark} of an orthogonal complement of $A$ \cite{candes2005decoding}, where the spark of a matrix is defined to be the \emph{smallest number of linearly dependent columns of it} \cite{donoho2003optimally}. Specifically, for $A$ with a full column rank, we can find an orthogonal complement $A^{\bot}$ of it, 
and \eqref{cosparkprob} is equivalent to 
\begin{align}
 \underset{x}{\text{minimize}} &~~ ||x||_0,  \label{sparkprob}\\
 \text{subject to} &~~ A^{\bot}x = 0, ~x \ne 0, \nn
\end{align}
and the optimal value of \eqref{sparkprob} is the spark of $A^{\bot}$, denoted by $\textit{spark($A^{\bot}$)}$. Computing spark is known to be NP hard \cite{tillmann2014computational}. 

The role of $cospark(A)$ has been studied in decoding under sparse measurement errors where $A$ is the coding matrix \cite{candes2005decoding}. In particular, $\frac{cospark(A)}{2}$ gives the maximum number of errors that an ideal decoder can tolerate for exact recovery. Closely related to this is the role of $spark(A^{\bot})$ in characterizing the ability to perform compressed sensing \cite{candes2005decoding} \cite{donoho2003optimally}. Spark is also related to notions such as mutual coherence  \cite{donoho2003optimally}\cite{gribonval2003sparse}  and Restrict Isometry Property (RIP) \cite{candes2005decoding} \cite{candes2006stable} 
which provide conditions under which sparse recovery can be performed using $l$-1 relaxation. 
Last but not least, in addition to its role in the sparse recovery literature, cospark \eqref{cosparkprob} also plays a central role in security problems in cyber-physical systems (see \cite{zhao2016minimum} among others). 


In this paper, we study the problem of computing the cospark of a matrix. Although it is proven that $\eqref{cosparkprob}$ is an NP hard problem, we show that the cospark that a matrix ``generically'' has can in fact be computed in polynomial time. Specifically, given the ``sparsity pattern'', (i.e., the locations of all the non-zero entries of $A$,) $cospark(A)$ equals, \emph{with probability one}, to a particular number which we termed the \emph{generic cospark} of $A$, if the non-zero entries of $A$ are drawn from independent continuous probability distributions. Then, we develop an efficient algorithm that computes the generic cospark in \emph{polynomial time}. 
 
\section{Preliminaries}  \label{sec:prelim}
\subsection{Generic Rank of a Matrix}
For a matrix $A\in \mbb{R}^{m\times n}$, we define its $\textit{sparsity pattern}$ $S = \{ (i,j) | A_{ij} \ne 0, 1 \le i \le m, 1 \le j \le n \}$. Given a sparsity pattern $S$, we denote $A^{S}$ to be the set of all matrices with sparsity pattern $S$ over the field $\mathbb{R}$. Since there is a one to one mapping between $S$ and $A^{S}$, we use $S$ and $A^{S}$ interchangeably to denote a sparsity pattern in the remainder of the paper. 

The generic rank of a matrix with sparsity pattern $S$ is defined as follows. 

\begin{definition}[Generic Rank]\label{def:grank}
Given $S$, the \textit{generic rank} of $A^{S}$ is $sprank(A^{S}) \triangleq \sup_{A \in A_{S}} rank(A)$. 
\end{definition}
Clearly, if $sprank(A^S) < n$, the optimal value of \eqref{cosparkprob} is zero. 
We will thus focus on the case $sprank(A^S) = n$ for the remainder of the paper.
 
The following lemma states that the generic rank indeed ``generically'' equals to the rank of a matrix \cite{sprankbook}. 
\begin{lemma} \label{lem:grank}
Given $S$, $rank(A) = sprank(A^S)$ with probability one, if the non-zero entries of $A$ are drawn from independently distributed continuous probability distributions. 
\end{lemma}

\subsection{Matching Theory Basics}
We now introduce some basics from classical matching theory \cite{diestel2016graph} which are necessary for us to introduce the results in the remainder of the paper. 

For a bipartite graph $G(X,Y,E)$, a subset of edges $\mathcal{N} \subseteq E$ is a $\textit{matching}$ if all the edges in $\mathcal{N}$ are vertex disjoint. A $\textit{max matching}$ from $X$ onto $Y$ 
is a matching with the maximum cardinality. A \emph{perfect matching from $X$ onto $Y$} is a max matching where every vertex in $Y$ is incident to an edge in the matching. 

Consider a (not necessarily maximum) matching $\mathcal{N}$. A vertex is called \textit{matched} if it is incident to some edge in $\mathcal{N}$, and \textit{unmatched} otherwise. An $\textit{alternating path}$ with respect to $\mathcal{N}$ is a path which alternates between using edges in $E \setminus \mathcal{N}$ and edges in $\mathcal{N}$, or vice versa. An $\textit{augmenting path}$ w.r.t $\mathcal{N}$ is an alternating path w.r.t. $\mathcal{N}$ which starts and ends at unmatched vertices. 
With an augmenting path $P$, 
it can be easily shown that the symmetric difference\footnote{The \emph{symmetric difference} of two sets $S_1$ and $S_2$ is defined as $S_1\oplus S_2 = \left(S_1 \cup S_2\right) \setminus \left(S_1\cap S_2\right)$.} 
$\mathcal{N} \oplus P$ gives a matching with size $|\mathcal{N}| + 1$. 


\subsection{Generic Rank as Max Matching}
We now introduce an equivalent definition of generic rank via matching theory. 
A sparsity pattern $A^S$ can be represented as a \emph{bipartite graph} as follows \cite{sprankbook}. Let $G(X,Y,E)$ be a bipartite graph whose a) vertices $X = \{1,2,\ldots,m\}$ correspond to all the row indices of $A^{S}$, b) vertices $Y = \{1,2,\ldots,n\}$ correspond to all the column indices of $A^{S}$, and c) edges in $E=S$ correspond to all the non-zero entries of $A^S$. Accordingly, we also denote the bipartite graph for a sparsity pattern $S$ by $G(X, Y, S)$. 

The following lemma states the equality between $sprank(A^{S})$ and the max matching of $G(X,Y,S)$ \cite{sprankbook}. 
\begin{lemma} \label{lem:matchgrank}
Given $G(X,Y,S)$, the generic rank $sprank(A^{S})$ equals to the cardinality of the maximum bipartite matching on $G$. 
\end{lemma}
Accordingly, finding a max matching on this graph using the Hopcroft-Karp algorithm allows us to find the generic rank with $\mathcal{O}(|S|\sqrt{m+n})$ complexity \cite{hopcroft1973n}.

\section{Generic Cospark}

Similarly to the supremum definition of generic rank (cf. Definition \ref{def:grank}), given the sparsity pattern of a matrix, we define \emph{generic cospark} as follows. 
\begin{definition} [Generic Cospark]
Given $S$, the $\textit{generic cospark}$ of $A^{S}$ is $spcospark(A^{S}) \triangleq \sup_{A \in A_{S}} cospark(A)$.
\end{definition}
In a spirit similar to the multiple interpretations of generic rank as in Section \ref{sec:prelim}, we provide a \emph{probabilistic} view and a \emph{matching theory} based view of generic cospark as follows. 

\subsection{Cospark Equals to Generic Cospark With Probability One}
For any $T \subset [m]$, let $A_T$ and $A^{S}_{T}$ represent the matrix $A$ and the set of matrices $A^{S}$ restricted to the rows $T$ respectively. A class of matrices which has cospark equal to generic cospark are those which satisfy the following property: \\

\begin{lemma} \label{lem:matppt}
Given any sparsity pattern $S$ so that $sprank(A^{S}) = n$ for $A^{S} \subset \mathbb{R}^{m \times n}$, 
for any $A\in A^S$, 
if $rank(A_T) = sprank(A^{S}_{T}), \forall T \subseteq [m]$, then $cospark(A) = spcospark(A^S)$.
\end{lemma}

\begin{proof}
Let $x^* = \argmin_{x \ne 0} ||Ax||_0$, and suppose $U = \{ i | a_i x^{*} = 0 \}$, where $a_i$ is the $i$th row of $A$. Since $A_Ux^{*} = 0$, $rank(A_U) < n$. Now consider another matrix $C \in \mathbb{R}^{m \times n}$ with sparsity pattern $S$. 
Since $rank(C_U) \le rank(A_U) = sprank(A^{S}_U) < n$, $\ker(C_U)$ is also nonempty, meaning there exists a nonzero vector $h \in \mathbb{R}^{n}$ such that $C_Uh = 0$. Because $A_{U^c}x^{*}$ has no zero entries, we also have $||C_{U^c}h||_0 \le ||A_{U^c}x^{*}||_0 = ||Ax^{*}||_0$. This means $||Ch||_0 = ||C_Uh||_0 + ||C_{U^c}h||_0 \le ||A_{U^c}x^{*}||_0 = ||Ax^{*}||_0$. Hence, if $\hat{x} = \argmin_{x \ne 0} ||Cx||_0$, it follows $cospark(C) = ||C\hat{x}||_0 \le ||Ch||_0 \le ||Ax^*||_0 = cospark(A)$, which proves the lemma.
\end{proof}
We note that the property $rank(A_T) = sprank(A^{S}_{T}), \forall T \subseteq [m]$ is known as the $\textit{matching property}$ of matrix $A$ \cite{mccormick1983combinatorial}. 

Now, we have the following theorem showing that the generic cospark indeed ``generically'' equals to the cospark. 

\begin{Theorem} \label{thm:gcospark}
Given $S$, $cospark(A) = spcospark(A^S)$ with probability one, if the non-zero entries of $A$ are drawn from independently distributed continuous probability distributions.
\end{Theorem}

\begin{proof}
If we have a matrix $A$ with sparsity pattern $S$ whose nonzeros are drawn from independent continuous distributions, then every submatrix of rows has rank equaling generic rank w. p. 1 (cf. Lemma \ref{lem:grank}). 
This immediately implies $cospark(A) = spcospark(A^S)$ w. p. 1 by Lemma \ref{lem:matppt}. 
\end{proof} 


\subsection{A Matching Theory based Definition of Generic Cospark}


Let $G(X,Y,S)$ be the  bipartite graph corresponding to $A^{S}\subseteq \mathbb{R}^{m \times n}$. 
For a subset of vertices $Z\subseteq X$, we define the induced subgraph $G(Z)$ as the bipartite graph $G(Z, N(Z), \{ (i,j) | i \in Z\ , j \in N(Z)\})$, where $N(Z)$ denotes the vertices in $Y$ adjacent to the set $Z$. 
$G(Z)$ is essentially a bipartite graph corresponding to 
the submatrix $A^{S}_{Z}$. 
We then have the following. 

\begin{lemma} \label{lem:opt}
Given G(X, Y, S), let $OPT \subset X$ be a largest subset such that the induced subgraph $G(OPT)$ has a max matching of size $n-1$. We have that 
$\textit{spcospark($A^S$)} = m - |OPT|$. 
\end{lemma}

The intuition behind this matching theory based definition of $\textit{spcospark($A^S$)}$ is the following. To find the sparsest vector in the image of $A$, it is equivalent to find a largest set of rows in $A$, $OPT$, which span an $n-1$ dimensional subspace. With such a subset $OPT$, we can find a vector $x^{*}$ that satisfies $A_{OPT}x^* = 0$, and it is clear that $x^{*} \in \argmin_{x \ne 0} {||Ax||_0}$. 
Furthermore, based on the equivalence between generic rank and max matching from Lemma \ref{lem:matchgrank}, we arrive at the matching theory based definition of generic cospark in Lemma \ref{lem:opt}. 


\section{Efficient Algorithm for Computing Generic Cospark}
In this section, we introduce an efficient algorithm that computes the \emph{generic} cospark. 
This algorithm is based on a greedy approach motivated by Lemma \ref{lem:opt}. 

Given $G(X,Y,S)$, for any size $n-1$ subset of vertices $W \subset Y$, we define $X_W = \{x \in X | N(x) \subseteq W \}$. In other words, $X_W$ is the index set of rows of $A^{S}$ with a \emph{zero} entry in the remaining coordinate  $v = Y \setminus W$. 

We use $X_W$ as a basis to construct a candidate solution for $OPT$. The idea is to add a maximal subset of vertices $B \subset X_{W}^{c}$ to $X_W$, such that $\overline{X_W} = X_W \cup B$ has a matching of size $n-1$ onto $Y$. Specifically, we keep adding vertices $t \in X_{W}^{c}$ to $B$ as long as the submatrix corresponding to the index set $X_W \cup B$ has generic rank no greater than $n-1$. 
The following lemma shows that adding a vertex to $B$ can only increase the generic rank of $X_W \cup B$ by \emph{at most one}. 

\begin{lemma} \label{lem:inc1}
Given $G(X,Y,S)$, $\forall Z \subset X$ and $u \in X \setminus Z$, 
$sprank(A^{S}_{Z \cup \{u\}}) \le sprank(A^{S}_{Z}) + 1$.
\end{lemma}

\begin{RK}
For a given $W$, depending on the order we visit the vertices in $X_{W}^{c}$, we could end up with different sets $B$, possibly of different sizes. 
However, we will prove that the optimal solution is recovered regardless. 
\end{RK}

$\overline{X_W}, \forall W$ are the candidate solutions for OPT, and we obtain the optimal solution by choosing the $\overline{X_W}$ with the largest cardinality, i.e. $X_f = \argmax_{W \subset Y, |W| = n-1} |\overline{X_W}|$. The generic cospark of $A^S$ then equals to $m - |{X_f}|$. 

The detailed algorithm is presented in Algorithm \ref{spcospark}.

\begin{algorithm}
\caption{Computing Generic Cospark}\label{spcospark}
\begin{algorithmic}[1]
\Procedure{spcospark}{$A^{S}$}
\State Initialization: Set $B = \emptyset, t = \emptyset$, and $X_f$ = $\emptyset$
\ForAll{ $W \subset Y$ of cardinality $n-1$} 
\State \parbox[t]{\dimexpr\linewidth-\algorithmicindent}{Scan through all $m$ vertices in $X$ \\ to find $X_W$ and let $T = X_{W}^{c}$ \strut}
\State Calculate $sprank(A^{S}_{X_W})$
\While {$sprank(A^{S}_{X_W \cup B \cup \{t\}}) \le n-1$}
\State Let $B = B \cup t$ \strut
\State \parbox[t]{\dimexpr\linewidth-\algorithmicindent}{Choose any element ${t}$ from $T$, and \\ set $ T = T \setminus t$}
\EndWhile \ and let $\overline{X_W} = X_{W} \cup B$
\If{$|X_f| < |\overline{X_W}|$}
\State Set $X_f = \overline{X_W}$
\EndIf
\State Set $B = \emptyset$
\EndFor
\State Return $X_f$, and $spcospark(A^S) = m - |{X_f}|$.
\EndProcedure
\end{algorithmic}
\end{algorithm}

\section{Proof of Optimality of Algorithm \ref{spcospark}} 
In this section, we prove that Algorithm \ref{spcospark} indeed solves the generic cospark. It is sufficient to prove that the set $X_f$ returned by the Algorithm satisfies the definition of $OPT$ in Lemma \ref{lem:opt}, i.e., $X_f$ is a subset of vertices of the largest size such that the induced subgraph $G(X_f)$ has a max matching of size $n-1$. 
Since $G(X_f)$ by construction has a max matching of size $n-1$, \emph{it is sufficient to prove that $X_f$ has the largest size, i.e., $|X_f| = |OPT|$}. 

To prove this, let us consider an optimal set $OPT\subset X$. We denote by $\mc{M}$ the set of $n-1$ edges of a max matching of $G(OPT)$. We denote by $W^{*} \subset Y$ the set of $n-1$ vertices in $Y$ corresponding to this max matching, and denote by $v = Y \setminus W^{*}$ the remaining vertex in $Y$. We will show that, starting with $W^*$, Algorithm \ref{spcospark} will return an $X_f$ such that $|X_f| \ge |OPT|$, and hence $|X_f| = |OPT|$. 
As the notations for this section are quite involved, an illustrative diagram is plotted in Figure \ref{fig:proof} to help clarify the proof procedure in the following.

\begin{figure}[t!]
\centering
\includegraphics[scale=0.60]{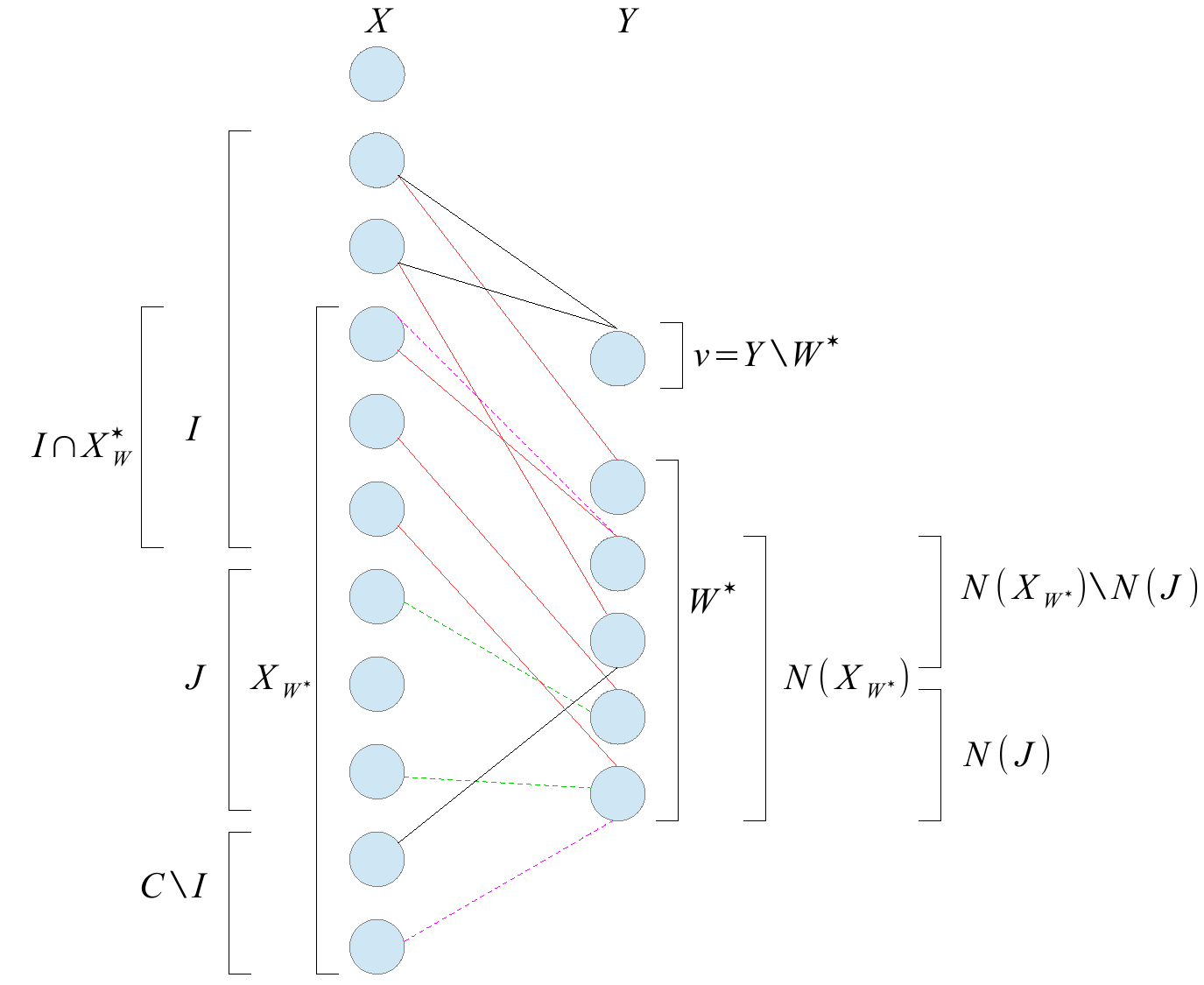} 
\caption{The above graph represents a sketch of the partition of the nodes. The black continuous line segments are unmatched edges in the bipartite graph. The red continuous line segments comprise $\mathcal{M}$, which forms a $n-1$ matching from $\mathcal{I}$ to $W^*$. The pair of pink dotted line segments denote the range of vertices $X_{W^*}$ is incident to. Finally, the pair of green dotted line segments denote the range of vertices $\mathcal{J}$ is incident to.} 
\label{fig:proof}
\end{figure}

We first partition $OPT$ into $OPT = \mathcal{I} \cup \mathcal{J}$, $\mc{I}\cap\mc{J} = \emptyset$, where $\mathcal{I}$ is the set of $n-1$ vertices in $OPT$ corresponding to the max matching $\mc{M}$. Hence, $\mc{I}$ perfectly matches onto $W^*$ with $\mc{M}$. $\mathcal{J}$ consists of the remaining vertices in $OPT$ unmatched by $\mc{M}$. 

WLOG, we assume $\mathcal{J}$ is nonempty. This is because, if $\mathcal{J}$ is empty, we then immediately have $|OPT| = n - 1 \le |X_f|$. 

We then have the following lemma about $\mc{I}$ and $\mc{J}$. 

\begin{lemma} \label{lem:JinXW}
For any such partition $OPT = \mathcal{I} \cup \mathcal{J}$, we have that $\mathcal{J} \subset X_{W^*}$, and $\mathcal{I} \cap X_{W^*}$ is nonempty.
\end{lemma}

\begin{proof}
Let $OPT$ be partitioned into $\mathcal{I} \cup \mathcal{J}$. Suppose $j \in \mathcal{J}$. If $j \notin X_{W^*}$, then $j$ is incident to $v$, which means $\mathcal{I} \cup \{j\}$ has a perfect matching onto $Y$. This contradicts $OPT$ has no perfect matching onto $Y$. Now suppose $\mathcal{I} \cap X_{W^*}$ is empty. This means every vertex in $\mathcal{I}$ is incident to $v$. Since $\mathcal{I}$ has a perfect matching onto $W^{*}$ and vertices in $\mathcal{J}$ are incident to vertices in $W^{*}$, it follows there exists an augmenting path from any vertex in $\mathcal{J}$ to $v$, which is a contradiction to $sprank(A^S_{OPT}) = n-1$. 
\end{proof}
Accordingly, we can partition $X_{W^*} = \mc{C} \cup\mc{J}, ~  \mc{C}\cap\mc{J} = \emptyset$, with $\mc{C}\triangleq X_{W^*} \setminus \mc{J}$.  
Starting from here, the general idea of proving $|X_f| \ge |OPT|$ is to lower bound 
\begin{align}
|X_f| = |X_{W^*}\cup B| = |\mc{C}\cup\mc{J}\cup B| =  |\mc{C}| + |\mc{J}| +  |B|.  \label{eq:tolowerbound}
\end{align}

We immediately have the following lower bound on $|B|$: 
\begin{align}
|B| \ge (n-1) - sprank(A^{S}_{X_{W^*}}). \label{eq:ineb}
\end{align}
This is because a) Algorithm \ref{spcospark} guarantees that $sprank(A^S_{W^*\cup B}) = n-1$, and b) every time we add a new vertex $t$ into $B$ (cf. Line 7 in Algorithm \ref{spcospark}), $sprank(A^S_{W^*\cup B})$ increases by at most one (cf. Lemma \ref{lem:inc1}). Since the initial generic cospark is $sprank(A^{S}_{X_{W^*}})$, we need at least $(n-1) - sprank(A^{S}_{X_{W^*}})$ vertices added into $B$ to reach $sprank(A^S_{W^*\cup B}) = n-1$. 

We next devote the majority of this section to provide a lower bound on $|\mc{C}|$.

\subsection{Lower Bounding $|\mc{C}|$}

The key result we will rely on in this subsection is the following: 
\begin{Theorem} \label{thm:noj}
For the induced bipartite graph $G(X_{W^*})$, there exists a max matching that does not touch any vertices in $\mathcal{J}$. 
\end{Theorem}

To prove Theorem \ref{thm:noj}, we start with a partial matching $\mathcal{M}_p \subset \mathcal{M}$ consisting only edges that touch $\mathcal{I} \cap X_{W^*}$. In other words, $\mathcal{M}_p = \{(i,j) \in \mathcal{M} | i \in \mathcal{I} \cap X_{W^*}\}$. 
The idea is that we will build a max matching starting from $\mc{M}_p$, and this max matching will not touch any vertices in $\mc{J}$, thus proving Theorem \ref{thm:noj}. 


We have the following two lemmas. 
\begin{lemma} \label{lem:njmat}
For the induced bipartite graph $G(X_{W^*})$ with $\mathcal{M}_p$ as a (not necessarily max) matching, any vertex in $N(\mathcal{J})$ is incident to some edge in $\mathcal{M}_p$, i.e., already matched. 
\end{lemma}

\begin{proof}
First, note any $j \in \mathcal{J}$ is not incident to $v$, so any vertex $k \in N(j)$ is in $W^*$. Now, for any $j \in \mathcal{J}$ and any $k \in N(j)$, we want to prove $k$ is incident to some edge in $\mathcal{M}_p$. Since $k \in W^*$, $k$ is incident to some edge $(i,k) \in \mathcal{M}$. $\mathcal{M}$ is the perfect matching from $\mathcal{I}$ to $W^*$, so certainly, $i$ is in $\mathcal{I}$. On the other hand, $i$ cannot be incident to $v$, or there will exist a length $3$ augmenting path from $j$ to $k$ to $i$ to $v$. Hence, $i \in X_W^*$, and the claim is proven. 
\end{proof}

\begin{lemma} \label{lem:nojtou}
For the induced bipartite graph $G(X_{W^*})$ with $\mathcal{M}_p$ as a (not necessarily max) matching, there exists no augmenting path starting from any $j \in \mathcal{J}$. 
\end{lemma}

\begin{proof}
For any vertex $ u \in N(X_{W^*}) \setminus N(\mathcal{J})$ that is unmatched w. r. t. $\mathcal{M}_p$, suppose there is an augmenting path from $j$ to $u$ using edges in $\mathcal{M}_p$. If $u$ is unmatched in the induced graph $G(X_{W^*})$ w.r.t $\mathcal{M}_p$ and $u \in W^*$, then there exists an edge $(i,u) \in \mathcal{M} \setminus \mathcal{M}_p$ which is incident to $u$. Because $(i,u) \in \mathcal{M} \setminus \mathcal{M}_p$, $i$ must be incident to $v$. This means if there exists an augmenting path from $j$ to $u$ w.r.t $\mathcal{M}_p$, then there must exist an augmenting path from $j$ to $v$ w.r.t $\mathcal{M}$, which contradicts vertices in $\mathcal{J}$ do not have augmenting paths to $v$. 
\end{proof}

Lemma \ref{lem:nojtou} implies that all augmenting paths w. r. t. the partial matching $\mathcal{M}_p$ are from unmatched vertices in $\mc{C} \setminus \mathcal{I}$ (where $\mc{C} =  X_{W^*} \setminus \mathcal{J}$) to unmatched vertices in $N(X_{W^*}) \setminus N(\mathcal{J})$. A corollary which will prove useful is the following:

\begin{corollary}
Suppose  $P$ is an augmenting path from $c \in \mc{C} \setminus \mathcal{I}$ to $u \in N(X_{W^*}) \setminus N(\mathcal{J})$ w. r. t. the matching $\mathcal{M}_p$. Then for any $j \in \mathcal{J}$, there exists no alternating path w. r. t. $\mathcal{M}_p$ from $j$ to any vertex in $P$. 
\end{corollary}

\begin{proof}
Let $P$ be an augmenting path from $c$ to $u$ w.r.t. $\mathcal{M}_p$. Suppose there exists an alternating path $P'_{jp}$ from $j$ to a vertex $p$, where $p$ is the first vertex in $P$ encountered when traversing $P'_{jp}$. $P'_{jp}$ must have odd number of edges, since $p$ is a matched vertex in $P$ and $j$ is unmatched. Since $P'_{jp}$ is odd, $p \in N(X_{W^*})$. Hence, if $P_{cp} \subset P$ is the restriction of $P$ from $c$ to $p$, then the alternating path $P_{cp}$ must also have odd length. The total length of $P$ must be odd since $P$ is an augmenting path, which means the length of the alternating path from $p$ to $u$ in $P$ must be even. 

Since $P_{jp}$ is an odd alternating path from $j$ to $p$, and the alternating path from $p$ to $u$ in $P$ is even, then the alternating path from $P_{jp}$ to $u$ is odd. Furthermore, $j$ and $u$ are unmatched,  so this path is actually an augmenting path, which immediately contradicts Lemma \ref{lem:nojtou}.
\end{proof}

From Corollary 1, any alternating path starting from $j$ w. r. t. $\mc{M}_p$ is \emph{vertex disjoint} to any augmenting path $P$. This implies that a) any alternating path from $j$ w. r. t. $\mathcal{M}_p \oplus P$ remains an alternating path, and b) there remains no augmenting path starting from $j$ w. r. t. $\mathcal{M}_p \oplus P$, i.e., \emph{Lemma \ref{lem:nojtou} continues to hold for $G(X_{W^*})$ with a new matching $\mathcal{M}_p \oplus P$}.

We are now ready to prove Theorem \ref{thm:noj}. 


\begin{proof}[Proof of Theorem \ref{thm:noj}]
Take $\mathcal{M}_p$ to be an initial matching onto $N(X_{W^*})$. By Lemma \ref{lem:njmat}, all vertices in $N(\mathcal{J})$ are now matched, and Lemma \ref{lem:nojtou} tells us we are left with augmenting paths starting from unmatched vertices in $\mc{C} \setminus \mathcal{I}$ to unmatched vertices in $N(X_{W^*}) \setminus N(\mathcal{J})$. If $P_1$ is one such augmenting path, then $\mathcal{M}_p \oplus P_1$ is a matching with one greater cardinality. By Corollary 1, all alternating paths w.r.t $\mathcal{M}_p$ starting from $j$ are vertex disjoint to $P_1$, which implies alternating paths starting from $j$ remain unchanged. Furthermore, Corollary 1 tells us $\mathcal{M}_p \oplus P_1$ does not have augmenting paths starting from $j$. Hence, the only remaining augmenting paths are still from vertices $\mc{C} \setminus \mathcal{I}$ to vertices $N(X_{W^*}) \setminus N(\mathcal{J})$. If $P_2$ is such an augmenting path, we can now repeat the above procedure and compute the matching $\mathcal{M}_p \oplus P_1 \oplus P_2$. Again, alternating paths starting from $j$ remain unchanged, and $\mathcal{M}_p \oplus P_1 \oplus P_2$ contains no augmenting paths starting from $j$. We can repeat this procedure until all augmenting paths from $\mc{C} \setminus \mathcal{I}$ to $N(X_{W^*}) \setminus N(\mathcal{J})$ are eliminated. Since the final matching obtained this way has no augmenting paths, this final matching is optimal, and its edges are incident to no vertices in $\mathcal{J}$. 
\end{proof}

As a result of Theorem \ref{thm:noj}, there exists a max matching of the bipartite graph $G(X_{W^*})$ that, on the ``left hand side'' of the graph, only touches vertices in $\mc{C} = X_{W^*}\setminus \mathcal{J}$. Since the size of the max matching of $G(X_{W^*})$ equals to $sprank\left(A^S_{X_{W^*}}\right)$ (cf. Lemma \ref{lem:matchgrank}), we arrive at the following lower bound on $|\mc{C}|$:
\begin{align}
|\mc{C}| \ge sprank\left(A^S_{X_{W^*}}\right). \label{eq:inec}
\end{align}

\subsection{Proof of the Optimality of Algorithm \ref{spcospark}}
We now show that Algorithm \ref{spcospark} indeed returns the generic cospark as in the following theorem. 
\begin{Theorem}
For the $X_f$ that Algorithm \ref{spcospark} returns, we have that $|X_f| = |OPT|$.
\end{Theorem}

\begin{proof}
By the definition of $OPT$, $|X_f| \le |OPT|$. To prove  $|X_f| \ge |OPT|$, starting from \eqref{eq:tolowerbound},  
\begin{align}
|X_f| & =  |\mc{C}| + |\mc{J}| +  |B| \\
& \ge sprank(A^{S}_{X_{W^*}}) + |\mathcal{J}| + |B| \label{eq:ine1} \\
& \ge sprank(A^{S}_{X_{W^*}}) + |\mathcal{J}| + (n-1) - sprank(A^{S}_{X_{W^*}}) \label{eq:ine2} \\
& = |\mathcal{J}| + (n-1) = |\mathcal{I}| + |\mathcal{J}| = |OPT|, 
\end{align}
where \eqref{eq:ine1} is from \eqref{eq:inec}, and \eqref{eq:ine2} is from \eqref{eq:ineb}. 
\end{proof}

\section{Algorithm Complexity}

We now show that Algorithm \ref{spcospark} is efficient, and provide an upper bound on its computational complexity. 

\begin{Theorem}\label{thm:comp}
Given any $S$, 
Algorithm \ref{spcospark} computes $spcospark(A^{S})$ in $\mathcal{O}(nm(1+|S|))$ time. 
\end{Theorem}

\begin{proof}
Observe in the pseudocode above, step 3 is over $n$ iterations. For each iteration, steps 4 to 9 are the most computationally expensive. Step 4 requires a $\mathcal{O}(m)$ scan of the rows of $A^{S}$, and step 5 requires us to compute a perfect matching using Hopcroft-Karp algorithm, which can be done in $\mathcal{O}(|S|\sqrt{m+n})$ time. 

For the loop in steps 6 to 9, we do not need to recalculate $sprank(A^{S}_{\{X_{W} \cup B \}})$ every iteration. Given we know the max matching from the previous iteration, we only need to check if the new vertex $t$ added to $B$ has an augmenting path to an unmatched vertex in $Y$. Searching for this augmented path requires us to use breadth first search (BFS) or depth first search (DFS), which can be computed in $\mathcal{O}(|S|)$ time. Since there are $\mathcal{O}(m)$ iterations in the while loop, the total cost of steps 6 to 9 is $\mathcal{O}(m|S|)$. 

Hence, for every iteration of step 3, the total cost is $\mathcal{O}(m + |S|\sqrt{m+n} + m|S|) =  \mathcal{O}(m(1+|S|))$ since $n \le m$. It follows immediately our total running time is $\mathcal{O}(nm(1+|S|))$.
\end{proof}

From Theorem \ref{thm:comp}, if $A^{S}$ is extremely sparse, the running time of Algorithm \ref{spcospark} is essentially \emph{quadratic}. 

\begin{RK}
The algorithm's bottleneck is in steps 6-9. For each row $t$ to add, we need to use a BFS. Since we need to add $\mathcal{O}(m)$ such vertices, the total complexity for these steps is $\mathcal{O}(m|S|)$ as in the above proof. To improve this complexity, we would like to detect multiple candidate rows to add to $B$ using a single BFS. Indeed, it can be shown further that steps 6-9 of Algorithm \ref{spcospark} can be improved to $\mathcal{O}(\sqrt{m}|S|)$ based on an idea similar to Hopcroft-Karp matching \cite{hopcroft1973n}. This will improve the total running time of Algorithm \ref{spcospark} to $\mathcal{O}(n\sqrt{m}|S|)$ . Details are omitted here. 
\end{RK}

\section{Experimental Results for Verification}

We compare the results from our algorithm of finding the \emph{generic cospark} to a brute force algorithm of finding the \emph{cospark}. 
Because the brute force algorithm has a computational complexity of $\mathcal{O}(m^n)$, we limit the size of the test matrices to $m = 20$ and $n = 5$. 

We run our comparison over 10 different sparsity levels spaced equally between zero and one. For each sparsity level, we generate 50 matrices, where the locations of the nonzero entries are chosen uniformly at random given the sparsity level, and the values of the non-zero entries are drawn from independent uniform distributions in $[0,1]$. For each of these 50 matrices, we compare the generic cospark given by Algorithm \ref{spcospark} versus that given by the brute force method. In every case, the solutions of both algorithms match. These results support the fact that our algorithm not only computes the generic cospark in polynomial time, but also obtains the actual cospark w. p. 1 if the non-zero entries are drawn from independent continuous probability distributions. 
%
%
%

\section{Conclusion}
We have shown that, although computing the cospark of a matrix is an NP hard problem, computing the generic cospark can be done in polynomial time. We have shown that, given any sparsity pattern of a matrix, the cospark is always upper bounded by the generic cospark, and is equal to the generic cospark with probability one if the nonzero entries of the matrix are drawn from independent continuous probability distributions. An efficient algorithm is developed that computes generic cospark in polynomial time. 


\bibliographystyle{IEEEtran}
\bibliography{refs}

\end{document}

%% file: defs.tex
\usepackage{ifpdf}
\usepackage{cite}
\usepackage{graphicx}
\usepackage{amsmath}
\usepackage[tight,footnotesize]{subfigure}
\usepackage{stfloats}
\usepackage{amssymb}
\usepackage{bm}

\newtheorem{RK}{Remark}

\newcommand{\nn}{\nonumber}
\newcommand{\bit}{\begin{itemize}}
\newcommand{\eit}{\end{itemize}}
\newcommand{\mbb}{\mathbb}

\usepackage{ifthen}
\newboolean{showcomments}
\setboolean{showcomments}{true}
\newcommand{\yue}[1]{\ifthenelse{\boolean{showcomments}}
{ \textcolor{red}{(Yue says:  #1)}}{}}

\DeclareMathOperator*{\argmax}{argmax}
\DeclareMathOperator*{\argmin}{argmin}

\def\mc{\mathcal}